\documentclass[12pt]{article}
\pdfoutput=1
\usepackage{amsmath}
\usepackage{amsfonts}
\usepackage{amssymb}
\usepackage{color}
\usepackage{authblk}
\usepackage{hyperref}
\usepackage{verbatim}
\usepackage{graphicx}
\usepackage{epsfig}
\usepackage{epsf}

\usepackage{amsthm}
\usepackage[small, margin=1cm]{caption}

\usepackage{soul,color}

\usepackage{bm}
\usepackage{multirow}
\usepackage[round]{natbib}

\makeatletter
\@addtoreset{equation}{section}
\makeatother

\newtheorem{theorem}{\bf Theorem}[section]

\newtheorem{algorithm}[theorem]{Algorithm}

\newtheorem{example}[theorem]{Example}

\newtheorem{proposition}[theorem]{Proposition}

\theoremstyle{plain}
\newtheorem{thm}{Theorem}[section]

\newtheorem{prop}[thm]{Proposition}
\theoremstyle{remark}

\newcommand{\E}[1]{{\bf E}\left(#1\right)}

\newcommand{\beq}{\begin{eqnarray}}
\newcommand{\beqs}{\begin{eqnarray*}}
\newcommand{\eeq}{\end{eqnarray}}
\newcommand{\eeqs}{\end{eqnarray*}}

\newcommand{\sign}{\mbox{sign}}
\newcommand{\median}{\mbox{median}}
\newcommand{\argmax}{\operatornamewithlimits{argmax}}

\newcommand{\remove}[1]{}

\title{A penalized likelihood approach for robust estimation of isoform expression}

\author[1,3,*]{Hui Jiang}
\author[2,4,*]{Julia Salzman}
\affil[1]{Department of Biostatistics, University of Michigan}
\affil[2]{Department of Biochemistry, Stanford University}
\affil[3]{Center for Computational Medicine and Bioinformatics, University of Michigan}
\affil[4]{Stanford Cancer Institute, Stanford University}
\affil[*]{Please send correspondence to jianghui@umich.edu and julia.salzman@stanford.edu.}

\begin{document}

\maketitle

\abstract{Ultra high-throughput sequencing of transcriptomes (RNA-Seq) has enabled the accurate estimation of gene expression at individual isoform level. However, systematic biases introduced during the sequencing and mapping processes as well as incompleteness of the transcript annotation databases may cause the estimates of isoform abundances to be unreliable, and in some cases, highly inaccurate. This paper introduces a penalized likelihood approach to detect and correct for such biases in a robust manner. Our model extends those previously proposed by introducing bias parameters for reads. An L1 penalty is used for the selection of non-zero bias parameters. We introduce an efficient algorithm for model fitting and analyze the statistical properties of the proposed model. Our experimental studies on both simulated and real datasets suggest that the model has the potential to improve isoform-specific gene expression estimates and identify incompletely annotated gene models.
}

\section{Introduction}

In eukaryotes, a single gene can often produce more than one distinct
transcript isoforms, through an important cell mechanism called
alternative splicing. Alternative splicing can greatly enrich the
diversity of eukaryote transcriptomes~\citep{Wang2008}, especially in
developmental and differentiation programs, and can contribute
to disease when it is dysregulated~\citep{LopezBigas2005}.  Study gene expression at specific transcript isoform level is therefore of great importance and interest to biologists.

Ultra high-throughput sequencing of transcriptomes (RNA-Seq) has enabled
the accurate estimation of gene expression at individual isoform
level~\citep{Wang2008}. As of today, modern ultra high-throughput
sequencing platforms can generate tens of millions of short sequencing reads from
prepared RNA samples in less than a day.  For these reasons, RNA-Seq has
become the method of choice for assays of gene expression. To analyze
increasing amounts of data generated from biological experiments, a number
of statistical models and software tools have been
developed~\citep{Jiang2009, Trapnell2010, Li2011}.  For a review of the
methods for transcript quantification using RNA-Seq, see~\citet{Pachter2011}.

Although these methods have achieved great success in quantifying isoforms
accurately, there are still many remaining challenging issues which may
hinder their wider adoption and successful application by
biologists. Systematic biases introduced during the sequencing and mapping
processes~\citep{Li2010, Hansen2010, Roberts2011} as well as
incompleteness of the transcript annotation databases~\citep{Pruitt2009, Hsu2006}
can cause the estimates of isoform abundances to be unreliable.  For
example, recently, there have been periods of time where hundreds
of new transcripts are discovered every month~\citep{Harrow2012},
including occasional examples of thousands of new isoforms being
identified in a single study, ~\citep{Salzman2012, Salzman2013}.  These
incomplete annotations
can cause the estimates of isoform abundances to be
unreliable~\citep{Pyrkosz2013}.

This paper introduces a penalized likelihood approach to detect and
correct for such biases in a robust manner.  Bias parameters are
introduced for read abundance, and an L1 penalty is used for the selection
of non-zero bias parameters.   We introduce an efficient algorithm for
fitting this model and analyze its statistical properties.  Our experimental studies on both simulated and real
  datasets show that transcript estimates can be highly sensitive to
  including or omitting parameters modeling read bias.  Together, our
  results suggest that this method has the potential
  to improve isoform-specific gene expression estimates and improve
  annotation of existing gene models.

\section{A penalized likelihood approach}

\subsection{The Model}\label{model}
We adopt the notation and extend the model in~\citet{Salzman2011}, which
provides a flexible statistical framework for modeling for both single-end
and paired-end RNA-Seq data, including insert length distributions. To
state the model, for a gene $g$ with $I$ annotated distinct transcript
isoforms, suppose the sequencing reads from $g$ are sampled from $J$
possible distinct read types. A read type refers to a group of reads
(single-end or paired-end) with the same probability of being generated by
sequencing a particular transcript~\citep{Salzman2011}.  We use $\theta$
to be the $I\times1$ vector representing the abudance of the isoforms in
the sample, $A$ to be the $I\times J$ sampling rate matrix with its $(i,j)$-th element $a_{ij}$ denoting the rate that read type $j$ is sampled from isoform $i$. Given $\theta$ and $A$, we assume that the $J\times1$ read count vector $n$, where $n_j$ denotes the number of reads of type $j$ mapped to any of the $I$ isoforms, follows a Poisson distribution
$$n_j|\theta,A\sim \text{Poisson}\left(\sum_{i=1}^I\theta_ia_{ij}\right).$$
The log-likelihood function is therefore
$$l(\theta;n,A)=\sum_{j=1}^J\left\{n_j\ln\left(\sum_{i=1}^I\theta_ia_{ij}\right)
-\sum_{i=1}^I\theta_ia_{ij}\right\},$$
where the term $-\ln(n_j!)$ was dropped because it does not contain $\theta$.

In ~\citep{Salzman2011}, the sampling rate matrix $A$ is a set of
parameters, assumed to be a
known function of the sequencing library and gene. For single-end RNA-Seq data, the simplest model is to assume the uniform sampling model which assigns $a_{ij}$ as $N$ where $N$ is the sequencing depth (proportional to total number of mapped reads) of the experiment if isoform $i$ can generate read type $j$ or assigns $a_{ij}$ as $0$ otherwise. For paired-end RNA-Seq data, an insert length model can be used such that $a_{ij}=q(l_{ij})N$ if read type $j$ can be mapped to isoform $i$ with insert length (fragment length) $l_{ij}$, where $q()$ is the empirical probability mass based on all the mapped read pairs. \citet{Salzman2011} discuss these sampling rate models in more details.

Although these simplified sampling rate models usually work well in
practice, there are systematic biases introduced during the sequencing and
mapping processes which may caused biased estimates of the sampling rates
and consequently biased estimates of isoform abundances.  Several
approaches have been developed to model and correct such
biases~\citep{Li2010, Hansen2010, Roberts2011}. However, completely
removing sampling biases is almost impossible because modeling and
identifying baiases in the technical
procedure of sequencing and read mapping is often too
complex.   Including all possible transcript isoforms (de novo identification) also poses
computational challenges and biases'; using all annotated transcripts in
the model, many times exceeeding 10 per gene, can introduce non-identifyability of
isoforms.  However, while the vast majority of human
genes have multiple annotated (and likely unannotated) transcripts, most
cell types, or single cells, express only a subset of annotated transcripts.

To explore statistical modeling approaches that could improve transcript
quantification with RNA-Seq, we present a flexible model to account for all
different kinds of biases in estimated sampling rates. We assign a bias
parameter $\beta_j$ to each read type $j$ and reparametrize $\beta_j$
as $\beta_j=e^{b_j}$ to constrain $\beta_j>0$.  When $\beta_j=1$, there is no bias for read type $j$. The actual effective sampling rate for read type $j$ from isoform $i$ now becomes $a^\prime_{ij}=a_{ij}\beta_j=a_{ij}e^{b_j}$, and the log-likelihood function is now
\begin{equation}\label{new-log-likelihood}
l(\theta,b;n,A)=\sum_{j=1}^J\left\{n_j\ln\left(\sum_{i=1}^I\theta_ia_{ij}e^{b_j}\right)
-\sum_{i=1}^I\theta_ia_{ij}e^{b_j}\right\}.
\end{equation}

Since the number of observations is $J$, which is smaller than the number
of variables $I$+$J$ in model~(\ref{new-log-likelihood}), the
model~(\ref{new-log-likelihood}) is not identifiable. To solve this problem, we introduce a penalty $p(b)$ on the bias parameters $b$ and formulate an L1-penalized log-likelihood
\begin{equation}\label{penalized-log-likelihood}
\begin{array}{lll}
f(\theta,b)&=&l(\theta,b;n,A)-p(b)\\
&=&\displaystyle{\sum_{j=1}^J\left\{n_j\ln\left(\sum_{i=1}^I\theta_ia_{ij}e^{b_j}\right)
-\sum_{i=1}^I\theta_ia_{ij}e^{b_j}\right\}-\lambda\sum_{j=1}^J|b_j|}
\end{array}
\end{equation}
where $\lambda>0$ is a tuning parameter.

Introducing the L1 penalty shrinks $b$
towards $0$, consequently inflating $\theta$ compared to fitting a model
with sparse
but unbiased estimation of the parameters $b$. One way to reduce such
bias in the estimation of $\theta$ is to use a two-step approach for model fitting: first fit the
model with the L1-penalized model, then fit the model without the L1
penalty, retaining only the non-zero $b_j$'s as model parameters. Clearly, to avoid
nonidentifiable issues, the number of non-zero $b_j$'s must be smaller or
equal than $J-I$, which can be achieved by increasing the tuning parameter $\lambda$. The statistical properties of the two-step approach is discussed in Section~\ref{statistical-properties}.

Because $J$ (the number of distinct read types) is usually very large,
especially for paired-end RNA-Seq data, we adopt the collapsing technique
introduced in~\citet{Salzman2011}. We merge read types of proportional
sampling rate vectors into read categories (which are minimal sufficient
statistics of the model). This does not change the
model~(\ref{penalized-log-likelihood}) except that $j$ now represents a read
category rather than a read type. \citet{Salzman2011} also introduced
another data reduction technique which ignores all the read categories
with zero read counts by introducing an additional term with the total sampling
rates for each isoform $w_i=\sum_{j=1}^Ja_{ij}$. In this case, the log-likelihood function becomes
\begin{equation}\label{log-likelihood2}
l(\theta;n,A,w)=\sum_{n_j>0}\left\{n_j\ln\left(\sum_{i=1}^I\theta_ia_{ij}\right)\right\}-\sum_{i=1}^I\theta_iw_i.
\end{equation}
For simplicity, we will not discuss model~(\ref{log-likelihood2}) in this
paper but our approach easily extends to deal with model~(\ref{log-likelihood2}).

\subsection{Optimization}

In this section we develop an efficient algorithm for fitting the model, i.e., maximizing the L1-penalized log-likelihood function~(\ref{penalized-log-likelihood}) $(\theta,b) = \argmax_{(\theta,b)}f(\theta,b)$.

\begin{proposition}\label{convexity}
	The L1-penalized log-likelihood function~$(\ref{penalized-log-likelihood})$ $f(\theta,b)$ is biconcave.
\end{proposition}

Because $f(\theta,b)$ is biconcave, we use Alternative Concave Search
(ACS) to solve for $\theta$ and $b$, by alternatively fixing one of them
and optimize for the other. The sequence of function values generated by
the ACS algorithm is monotonically increasing and $f(\theta, b)$ is
bounded from above (because it is a penalized log-likelihood function),
conditions which guarentee convergence of the ACS algorithm.

\begin{algorithm}\label{EM}
	With $b$ fixed, $\theta$ can be solved with the following EM algorithm
	\begin{eqnarray*}
		\mbox{E-step: }&\hat{n}_{ij}^{(k+1)}&\mathrel{\mathop:}=\E{n_{ij}|n,A,b,\theta^{(k)}}=\displaystyle{\frac{n_j\theta_i^{(k)}a_{ij}}{\sum_{i=1}^I\theta_i^{(k)}a_{ij}}}\\
		\mbox{M-step: }&\theta_i^{(k+1)}&=\displaystyle{\frac{\sum_{j=1}^J\hat{n}_{ij}^{(k+1)}}{\sum_{j=1}^Ja_{ij}e^{b_j}}}
	\end{eqnarray*}
\end{algorithm}

Alternatively, $\theta$ can be solved using the more efficient Newton-Raphson algorithm. In our implementation, we only execute one round of the EM iteration each time we optimize $\theta$ with $b$ fixed.

\begin{proposition}\label{solution-b}
With $\theta$ fixed, $b_j$ has the following closed-form solution
\begin{equation}\label{formula-b}
b_j=\ln\left(1+\frac{S_\lambda\left(n_j-\sum_{i=1}^I\theta_ia_{ij}\right)}{\sum_{i=1}^I\theta_ia_{ij}}\right)
\end{equation}
where $S_\lambda(x)=\sign(x)(|x|-\lambda)_+$ is the soft thresholding operator, where $(x)_+=\max(x,0)$.
\end{proposition}

For more efficient convergence, we use an analytical
property of values of $\theta$ and $b$ which maximize $f(\theta,b)$:

\begin{proposition}\label{median}
There is at least one set of $\theta$ and $b$ which maximize $f(\theta,b)$ such that $\median(b_1,\ldots,b_J)=0$.
\end{proposition}

Accordingly, after each iteration which solves ($\ref{formula-b}$), our
approach includes a step which centers the $b$'s around their median by
updating as follows:
$b^\prime_j=b_j-\median(b_1,\ldots,b_J)$.

\subsection{Statistical Properties}

Several statistical properties of the two-step approach introduced in
Section~\ref{model} are provided in this section. First, we state an
intuitive interpretation of the procedure:

\begin{proposition}\label{statistical-properties}
	Fitting the model using the two-step approach introduced in Section~\ref{model} is equivalent to fitting the model after removing all the observation $n_j$'s whose corresponding $b_j$'s are non-zero. In other words, the two model-fitting steps essentially perform outlier detection and removal, respectively.
\end{proposition}

This observation can be extended to prove that in the case of $I=1$, the
two-step procedure yields a consistent estimate of $\theta$ under the
assumption that each $b_j\geq 0$. While $I=1$ may appear to be a trivial
case, in fact, this approach is equivalent to considering a subset of the
full model introduced in~\citet{Salzman2011} where each read is considered if and only if it
can be generated by exactly one isoform.  Reasonable statistical
power can be achieved with this approach, and it is of relatively wide use
by biologists.  

For convenience, the proposition and proof is stated for the case
where $a_{1j}=N$, but this assumption
can be relaxed to allow $a_{1j}$ to be arbitrary.  Also, from the proof, it is clear that
larger choices of $\lambda$ than $(\max_jn_j)^{1/2}$ also result in
consistent estimates, but perhaps unnecessarily sparse models.

\begin{proposition}
\label{consistent}
Under the assumptions that $I=1$, $a_{1j}=N$ for all $1\leq j \leq J$,
$\lambda=(\max_jn_j)^{1/2}$, the two-step approach yields a consistent estimator of $\theta$.
\end{proposition}
\section{Experiments}
\subsection{Simulations}

In this section, we use simulation to study our model in various gene
structures, relative isoform abundances, bias pattern throughout
the gene and sequencing depths. For each simulation replicate, we estimate
$\theta$ and $b$ using three approaches and compare their estimation
accuracies. Throughout the simulations, we choose
$\lambda=(\max_jn_j)^{1/2}$ because of the consistency results we obtain
with this choice.
\begin{enumerate}
\item The conventional approach~\citep{Jiang2009,Salzman2011} with no bias correction (i.e., fix $b=0$).
\item Our proposed one-step approach with bias correction (i.e., without doing the second step of estimation $\theta$ introduced in Section~\ref{model}).
\item Our proposed two-step approach with bias correction.
\end{enumerate}

\begin{example}\label{example1}
We first simulate the case that a gene has a single annotated isoform (i.e., $I=1$), $5$ read categories after collapsing (i.e., $J=5$, e.g., the gene has $5$ exons). Suppose the estimate sampling rate matrix $A=NC$, where $C=(1,1,1,1,1)$ are the relative sampling rates for the five exons (e.g., each exon has the same length of 1000 bp), and $N$ is the relative sequencing depth (e.g., $N=10$ in Table~\ref{simulation1} means there are 10M single end reads from the sequencing experiment.  We assume $\theta=1$ and $b=(2,0,0,0,0)^T$.
\end{example}

We simulate 100 replicates and report the average (and standard deviation)
of estimation error of $\theta$ in L2 distance in
Table~\ref{simulation1}. We also report the average (and standard
deviation) of the number of $b$'s that are misidentified as zero vs. non-zero.
 Table~\ref{simulation1} shows empirical results confirming our
theory: if some $b_j>0$, without bias correction, $\theta$ will not be
estimated consistently. While both one-step and two-step approaches
achieve consistent estimates of $\theta$, the two-step approach is more
efficient.  On average, we misidentify less than one
nonzero $b$'s.

\begin{table}[ht]
\centering
\caption{Estimation accuracy of Example~\ref{example1}. Average of 100 replicates, standard deviation reported in parentheses.
\label{simulation1}}
\begin{tabular}{rlllll}
  \hline
Seq Depth & No Bias & Bias (1-step) & Bias (2-step) & \#Misientified \\
  \hline
10 & 1.32 (0.2) & 0.24 (0.14) & 0.13 (0.1) & 0.03 (0.17) \\
  100 & 1.26 (0.06) & 0.07 (0.04) & 0.04 (0.04) & 0.09 (0.29) \\
  1000 & 1.28 (0.02) & 0.02 (0.01) & 0.01 (0.01) & 0.05 (0.22) \\
   \hline
\end{tabular}
\end{table}

\begin{example}\label{example2}
We now consider a case with $I=2$, $J=6$ and $C=(1,2,1,2,3,2;1,2,0,2,3,2)$, e.g., a gene with six exons and two isoforms differ by the inclusion/exclusion of the third exon. We assume $\theta=(6,3)^T$ and $b=(-5,0,0)^T$.
\end{example}

The simulation results for Example~\ref{example2} are shown in Table~\ref{simulation2}. The performance of the three approaches is similar to that in Example~\ref{example1}.

\begin{table}[ht]
\centering
\caption{Estimation accuracy of Example~\ref{example2}. Average of 100 replicates, standard deviation reported in parentheses.
\label{simulation2}}
\begin{tabular}{rlllll}
  \hline
Seq Depth & No Bias & Bias (1-step) & Bias (2-step) & \#Misientified \\
  \hline
10 & 3.78 (0.23) & 3.41 (1.85) & 3.02 (1.71) & 0.13 (0.34) \\
  100 & 3.76 (0.08) & 1.82 (1.28) & 1.4 (0.9) & 0.01 (0.1) \\
  1000 & 3.77 (0.03) & 0.45 (0.32) & 0.36 (0.26) & 0 (0) \\
   \hline
\end{tabular}
\end{table}

\begin{example}\label{example3}
We now consider a case with $I=5$, $J=20$. For each replicate of the simulation, we randomly sample each element of $C$ as $c_{ij}=I_{u_1<0.1}0+I_{u_1>=0.1}\text{Uniform}(0,1)$, where $u_1\sim\text{Uniform}(0,1)$. We also randomly sample each element of $\theta$ and $b$ as $\theta_i\sim\text{Exponential}(1)$ and $b_i=I_{u_2<0.9}0+I_{u_2>=0.9}\text{N}(0,3)$ where  $u_2\sim\text{Uniform(0,1)}$.
\end{example}

The simulation results for Example~\ref{example3} are shown in
Table~\ref{simulation3}. The performance of the three approaches is
similar to that in Examples~\ref{example1} and~\ref{example2}. In particular, the approach without bias correction introduces huge estimation error in some of the cases (e.g., when $b_j$ is large and $a_{ij}$ is small).

\begin{table}[ht]
\centering
\caption{Estimation accuracy of Example~\ref{example3}. Average of 100 replicates, standard deviation reported in parentheses.
\label{simulation3}}
\begin{tabular}{rlllll}
  \hline
Seq Depth & No Bias & Bias (1-step) & Bias (2-step) & \#Misientified \\
  \hline
10 & 76.8 (471.72) & 1.22 (1.42) & 0.93 (0.69) & 2.17 (1.53) \\
  100 & 7792.35 (75161.75) & 2.56 (17.7) & 0.41 (0.41) & 2.2 (1.57) \\
  1000 & 406.35 (1934.52) & 0.42 (1) & 0.18 (0.41) & 2 (1.68) \\
   \hline
\end{tabular}
\end{table}

\subsection{Real data analysis}

We evaluted our model using real RNA-Seq data from the Gm12878 cell line
generated by the
ENCODE project~\citep{ENCODE}. A total of $415,630$ single-end reads of 75
bp mapped to human chromosome 22 are used in the analysis. We use the
human RefSeq annotation database~\citep{Pruitt2009} for our analysis. We
ran both the conventional approach (without bias correction) and our
proposed one-step approach (with bias correction) on this data set. 579
genes have estimated expression level $\geq1$ using the unit
RPKM~\citep{Mortazavi2008}, and 65 of the 579 genes have at least 2-fold change in their gene expression estimates between the approaches with and without bias correction.

MED15 is an example of a gene with greater than 2-fold change in the total
and expression of two isoforms of the gene
with and without bias correction, shown in Figure~\ref{example}.  The
center part of the gene has a much greater read density than the 5' or 3'
ends. Without bias correction, MED15's expression is estimated as 1487.11
RPKM (with the two isoforms estimated as 54.89 RPKM and 1432.22 RPKM,
respectively). The one-step approach identifies this bias and
down-weights the contribution of reads from center part of the gene to
total gene expression.  Consequently, it estimates the
gene expression as 702.52 RPKM (with the two isoforms estimated as 53.65 RPKM
and 648.87 RPKM, respectively).

The observed bias in MED15 could be due to mapping artifacts, or
preferential amplification of portions of the gene during RNA-Seq library
preparation.  Further investigation, including experimental testing may be
required to determine if either of these explanations for increased read
density are explanatory.  Another explanation could be that the gene
model used for our model, which includes just two isoforms, is incomplete.  For example, the observed increased
read density could be due to expression of other isoforms of MED15
that include these regions.

\begin{figure}
\begin{center}
\includegraphics[scale=0.75]{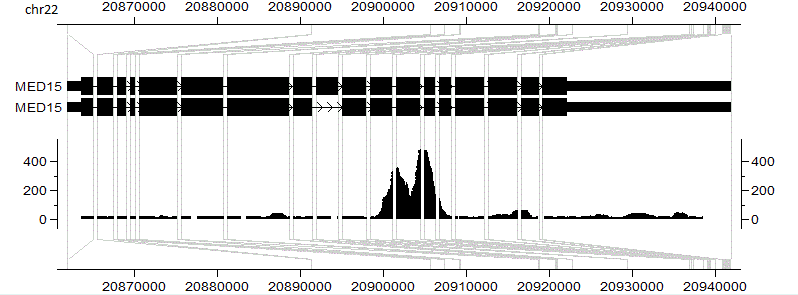}
\end{center}
\caption{Visualization of RNA-Seq reads mapped to the gene MED15 on human chromosome 22 in the CisGenome Browser~\citep{Jiang2010}. From top to bottom: genomic coordinates, gene structure where exons are magnified for better visualization, coverage of mapped reads. Reads are 75 bp single-end.}\label{example}
\end{figure}

\section{Discussion}

In this paper we choose $\lambda=(\max_jn_j)^{1/2}$, which seems to work
reasonably well with both simulated and real data and which we have shown
to produce consistent estimates of $\theta$ under reasonable
assumptions. We believe that more research on statistical
properites of different choices of $\lambda$ may lead to
improvement of our model in applied settings.  For example, we plan to
evaluate a standard approach of choosing $\lambda$ by
cross-validation, although it comes at the cost of more intensive
computation.  Also, as our proof of consistency shows, choosing values of
$\lambda$ larger than $(\max_jn_j)^{1/2}$ will also yield consistent
estimators of $\theta$ under the regime analyzed in Proposition $\ref{consistent}$.

 From~(\ref{formula-b}), it is apparent that the larger the value of $n_j$, the
 relatively smaller portion of it is affected by $\lambda$. Intuitively,
 the proposed approach works the best when the read categories are of
 similar sizes. In our real data experiment, we collapsed reads into exons
 and junctions to roughly fulfill this condition, and simulation
 demonstrates that our
 proposed approach is not very sensitive to how collapsing is performed. An alternative approach is to use $n_j$ as
 the weight for the corresponding $b_j$, i.e., by letting
 $p(b)=\lambda\sum_{j=1}^Jn_j|b_j|$. All the statistical properties and
 optimization techniques introduced in the paper can be adapted to this
 new penalty function with only minor modifications. In
 simulations (results not shown here), this new penalty function does not perform
 noticeably better than the current penalty function.  There may be
 advantages and disadvantages to increasing penalties for biases
 corresponding to larger $n_j$.

Although the two-step approach appears to be slightly more efficient than
the one-step approach in our simulations, it has several critical
drawbacks: 1) It requires an increase in computation up to a factor of
two; 2) It may introduce non-identifiable issue in the second step of
estimation when the number of nonzero $b$'s identified in the first step
of estimation is large; and 3) It makes parameter estimates sensitive to $\lambda$. Therefore, we use the one-step approach in our real data experiment and we plan to study the two-step approach in more details in future work.

The example of MED15 highlights another use of fitting bias
parameters. First, in the presence of unannotated isoforms of a gene,
correcting for bias in read sampling may be correcting for real biological
confounding.  In such scenarios, simulation suggests that correcting for
bias improves model fit and quantification conditional on the gene models used
for the study.  For example, the two transcripts of MED15 in Figure 1 are
probably more realistically estimated by our bias-corrected model.
In addition, screening genes with large estimated bias parameters may be
a tool for identifying unannotated transcripts or incomplete models used
in the mapping step.

Finally, the approach introduced in this paper is adapted and stated for the isoform
expression estimation problem, which is formally a Poisson regression
model with identity link function.  We believe it may be possible to generalize it to other generalized linear models such as linear regression models and logistic regression models, for which other practical applications may exist as well.

\section*{Appendix}

\begin{proof}[Proof of Proposition~\ref{convexity}]
\begin{eqnarray*}
f(\theta,b)&=&\displaystyle{\sum_{j=1}^J\left\{n_j\ln\left(\sum_{i=1}^I\theta_ia_{ij}e^{b_j}\right)
-\sum_{i=1}^I\theta_ia_{ij}e^{b_j}\right\}-\lambda\sum_{j=1}^J|b_j|}\\
&=&\displaystyle{\sum_{j=1}^J\left\{n_j\ln\left(\sum_{i=1}^I\theta_ia_{ij}\right)\right\}+\sum_{j=1}^Jn_jb_j
-\sum_{j=1}^J\sum_{i=1}^I\theta_ia_{ij}e^{b_j}-\sum_{j=1}^J\lambda|b_j|}
\end{eqnarray*}
where $n_jb_j$ and $-\lambda|b_j|$ are concave, $n_j\ln\left(\sum_{i=1}^I\theta_ia_{ij}\right)$ is concave because $-\sum_{i=1}^I\theta_ia_{ij}$ is concave and $\ln()$ is concave and non-decreasing, and $-\theta_ia_{ij}e^{b_j}$ is biconcave because both $\theta_i$ and $e^{b_j}$ are convex.
\end{proof}

\begin{proof}[Proof of Proposition~\ref{solution-b}]
Fixing $\theta$, since the L1-penalty is decomposable, $f(b)$ can be written as the sum of $J$ terms $f(b)=\sum_{j=1}^Jf_j(\theta,b_j)$, where
\begin{eqnarray*}
f_j(b_j)&=&\displaystyle{n_j\ln\left(\sum_{i=1}^I\theta_ia_{ij}e^{b_j}\right)
-\sum_{i=1}^I\theta_ia_{ij}e^{b_j}-\lambda|b_j|}\\
&=&n_jb_j+\displaystyle{n_j\ln\left(\sum_{i=1}^I\theta_ia_{ij}\right)
-e^{b_j}\sum_{i=1}^I\theta_ia_{ij}-\lambda|b_j|}\\
\end{eqnarray*}
Therefore, $b_j=\argmax_{b_j}f_j(\theta,b_j)$. Note the second term of $f_j(b_j)$ does not contain $b_j$. Since $|\cdot|$ is non-differentiable, we take the subdifferential of $f_j$ at $b_j$
$$\partial f_j(b_j)=n_j-e^{b_j}\sum_{i=1}^I\theta_ia_{ij}-\lambda s_j$$
where $s_j=\sign(b_j)$ if $b_j\neq0$ and $s_j\in[-1,1]$ if $b_j=0$. It can be verified that (\ref{formula-b}) is the solution to the equation $\partial f_j(b_j)=0$.
\end{proof}

\begin{proof}[Proof of Proposition~\ref{median}]
Suppose $\theta$ and $b$ are such that $(\theta,b)=\argmax_{(\theta,b)}f(\theta,b)$. Let $b_j^\prime=b_j-m$ and $\theta_i^\prime=\theta_ie^m$, where $m=\median(b_1,\ldots,b_J)$, then
\begin{eqnarray*}
f(\theta^\prime,b^\prime)&=&\displaystyle{\sum_{j=1}^J\left\{n_j\ln\left(\sum_{i=1}^I\theta_i^\prime a_{ij}e^{b_j^\prime}\right)
-\sum_{i=1}^I\theta_i^\prime a_{ij}e^{b_j^\prime}\right\}-\lambda\sum_{j=1}^J|b_j^\prime|}\\
&=&\displaystyle{\sum_{j=1}^J\left\{n_j\ln\left(\sum_{i=1}^I\theta_ie^ma_{ij}e^{b_j-m}\right)
-\sum_{i=1}^I\theta_ie^ma_{ij}e^{b_j-m}\right\}-\lambda\sum_{j=1}^J|b_j-m|}\\
&=&\displaystyle{\sum_{j=1}^J\left\{n_j\ln\left(\sum_{i=1}^I\theta_i a_{ij}e^{b_j}\right)
-\sum_{i=1}^I\theta_ia_{ij}e^{b_j}\right\}-\lambda\sum_{j=1}^J|b_j-m|}\\
&\geq&\displaystyle{\sum_{j=1}^J\left\{n_j\ln\left(\sum_{i=1}^I\theta_i a_{ij}e^{b_j}\right)
-\sum_{i=1}^I\theta_ia_{ij}e^{b_j}\right\}-\lambda\sum_{j=1}^J|b_j|}\\
&=&f(\theta,b)
\end{eqnarray*}
\end{proof}

\begin{proof}[Proof of Proposition~\ref{statistical-properties}]
Without log of generality, suppose $b_j\neq0, (j=1,\ldots,k)$ and $b_j=0, (j=k+1,\ldots,J)$ after the first step of model fitting with the L1 penalty. In the second step of model fitting without the L1 penalty, we have $(\theta,b)=\argmax_{(\theta,b)}f(\theta,b)$, where
\begin{equation}\label{f-new}
f(\theta,b)=\displaystyle{\sum_{j=1}^k\left\{n_j\ln\left(\sum_{i=1}^I\theta_ia_{ij}e^{b_j}\right)
-\sum_{i=1}^I\theta_ia_{ij}e^{b_j}\right\} +\sum_{j=k+1}^J\left\{n_j\ln\left(\sum_{i=1}^I\theta_ia_{ij}\right)
-\sum_{i=1}^I\theta_ia_{ij}\right\}}.
\end{equation}
Solving
$$\frac{\partial f(\theta,b)}{\partial b_j}=n_j-\sum_{i=1}^I\theta_ia_{ij}e^{b_j}=0$$
we have
\begin{equation}\label{solution-b-new}
b_j=\log\left(\frac{n_j}{\sum_{i=1}^I\theta_ia_{ij}}\right).
\end{equation}
Plugging~(\ref{solution-b-new}) into~(\ref{f-new}), we have
$$f(\theta,b)=\sum_{j=1}^k(n_j\ln n_j-n_j)+\sum_{j=k+1}^J\left\{n_j\ln\left(\sum_{i=1}^I\theta_ia_{ij}\right)
-\sum_{i=1}^I\theta_ia_{ij}\right\}$$
therefore $\theta=\argmax_\theta f^\prime(\theta)$ where
$$f^\prime(\theta)=\sum_{j=k+1}^J\left\{n_j\ln\left(\sum_{i=1}^I\theta_ia_{ij}\right)
-\sum_{i=1}^I\theta_ia_{ij}\right\}$$
is exactly the log-likelihood after removing all the observation $n_j$'s whose corresponding $b_j$'s are non-zero.
\end{proof}

\subsection*{Proof of consistency}

In this section, we assume $a_{1j}=N$ for convenience and prove
Proposition $\ref{consistent}$, that is, the two-step estimation procedure
is consistent, in several steps.
\begin{prop}
\label{max}
For $I=1$, any solution that maximizes $(2.2) $ under the constraint that $b_j \geq 0$ for $1 \leq j
\leq J$ must be a fixed point of the E-M algorithm and therefore, plugging
in to Equation~$(\ref{formula-b})$ implies that if $\hat{b}_j>0$,

\beq
\hat{\theta}e^{\hat{b}_j} = \frac{n_j - \lambda}N.
\eeq
Therefore, on the event for some $A$ a subset of $\{1,\ldots,
J\}$ where A is defined by $\hat{b_j} =0$ if and only if $j \in A$, the penalized log-likelihood reduces to:

\begin{equation}\label{f-subset}
f(\theta,b)=\displaystyle{ \sum_{j\in A}
  \left\{n_j\ln\left(\theta N
 \right) -\theta N\right\} -\sum_{j \notin A} \left\{(n_j - \lambda) - n_j\ln(n_j - \lambda ) \right\}  - \lambda \sum_{j
    \notin A} b_j}
\end{equation}

\end{prop}

\begin{proof}[Proof of $\ref{max}$]
The proof is by contradiction: if the solution does not maximize these
equations, iterating 1 step of the E or M will increase the likelihood.
Plugging in the solution for $b_j$ in terms of $\theta$ as in the proof of $(2.5)$.
\end{proof}

\begin{prop}[Corollary of Proposition $\ref{median}$]
\label{one}Under the constraint $b_j \geq 0$ for all $1 \leq j \leq J$,
the solution maximizing the likelihood ($\ref{penalized-log-likelihood}$) must estimate at least one $b_j=0$.
\end{prop}

Suppose $n_j$ is  $Po(N\mu_j)$ where $N$ is the ``sequencing depth'' and $\mu_j$
does not depend on $N$, and $\mu_j = e^{b_j} \theta$ and $$\hat{\mu}_j=\frac{n_j}{N}.$$  Limit results are stated for the case where $N \rightarrow \infty$.
In model ($\ref{penalized-log-likelihood}$), assume the $b_j$ are non-decreasing in $j$. Note
that $j\geq 1$ exists by Prop. $\ref{one}$.

\begin{prop}
Consider the model in ($\ref{penalized-log-likelihood}$) under a further assumption that for some
fixed $k$ where
$1\leq k \leq J$, $b_j=0$ for all $ 1 \leq j\leq k$ and $b_j >0$ for all $k<j\leq
J$.  The maximum likelihood estimator takes the form
$$ \hat{\theta}= \frac{\sum_{j \in A} n_j}{|A|N}$$ where $A$ indexes the subset
of $\{b_j\}_{1 \leq j \leq J}$ estimated to be zero, and $b_j$ takes the
form $(\ref{formula-b})$ for all $1 \leq j \leq J$.
\end{prop}

\begin{prop}
\label{consistency}
If $1(A \subset \{1, \ldots, k\}) \rightarrow 1$ a.s., where $1(\cdot)$ is the indicator function, on this event, $\hat{\theta}
\rightarrow \theta$ a.s.
\end{prop}
\begin{proof}[Proof of Prop. $\ref{consistency}$]
The first part is proved below. From the above equation $\ref{f-subset}$, it is clear that $1(A \subset \{1, \ldots, k\}) \rightarrow 1$ a.s. implies $\hat{\theta}
\rightarrow \theta$ a.s.
\end{proof}

Define the following procedure with tuning parameter $\lambda$ and the order statistics
of observed counts $n_{(1)},n_{(2)}, \ldots , n_{(J)}$: Intuitively,
Equation $\ref{f-subset}$ together with the closed form solution for $b_j$
in terms of an estimate of $\theta$ show that a fixed point of the EM
algorithm will define the set of variables with non-zero $b_j$ as those
where, omitting them,  the estimate of $N\theta$ plus $\lambda$ is less
than any observed $n_j$ ommitted.

To see this formally,
for each $1 \leq j \leq J$, let $\hat{\theta}_j = \frac{1}{jN}\sum_{i=1}^j n_{(i)}$.
Define $t$ as the smallest index $1\leq t\leq J-1$ such that $n_{(t+1)}> N\hat{\theta}_t +\lambda$. 
If no such $t$ exists, define $t=0$.

\begin{prop}
Suppose for some fixed $k$ where $1\leq k <J$, $b_j=0$ for all $ 1 \leq j\leq k$ and $b_j >
0$ for all $k< j \leq J$, i.e., $(b_1, \ldots, b_{k})= (0,\ldots, 0)$ and $b_{k+1}>0$  Fix $\lambda=(n_{(J)})^{1/2}$.  Then, as
$N \rightarrow \infty$, $1(1 \leq t \leq k) \rightarrow 1 $ a.s.
\end{prop}

\begin{proof}
Fix $i, j$ where  $i \leq k$ and $j>k$, ie. $i$ corresponds to a
$b_i=0$ and $j$ corresponds to a $b_j>0$. The CLT implies $1(n_i>n_j) \rightarrow 0$ a.s.
 as $N\rightarrow \infty$.

Therefore, summing over all $i$ and $j$ satisfying $1 \leq i \leq k,k<j \leq J$,
$$ \sum_{1 \leq i \leq k,k<j \leq J} 1(n_i>n_j) \rightarrow 0
\mbox{ a.s.}$$

It follows that for any $i$ and $j$ with $i \leq k$ and $j>k$,

$$1(n_{(j)} = n_i) \rightarrow 0 \mbox{ a.s.}$$

and hence that for any constant $c>0$ not growing with $N$,

\begin{enumerate}
\item For all $ 1 \leq i  \leq k$,
$$\hat{\mu}_i \rightarrow \theta \mbox{ a.s.}$$ and
 \item For all $ k< j  \leq J$ (in particular, for $j=k+1$)
$$1(\hat{\mu}_j > \theta + \frac{c}{\sqrt{N}})
   \rightarrow 1 \mbox{ a.s.}$$
\end{enumerate}
which implies that
$$1( N\hat{\mu}_j > N \theta + \sqrt{N \theta e^{b_J}})\rightarrow 1 \mbox { a.s.}$$
Therefore,
$$1( n_{(j)} > N\hat{\theta}_{k} + \sqrt{N \hat{\mu}_{(J)}})\rightarrow 1 \mbox { a.s.  }$$
That is, $$1(n_{(k+1)} > N\hat{\theta}_k+\lambda )
\rightarrow 1 \mbox { a.s. }$$
which implies $$1(1\leq t \leq k) \rightarrow 1  \mbox{ a.s.}$$
\end{proof}

\section*{Acknowledgements}
HJ's research was supported in part by an NIH grant 5U54CA163059-02 and a
GAPPS Grant from the Bill \& Melinda Gates Foundation. JS was supported by
NIH grant 1K99CA16898701 from the NCI.

\bibliographystyle{abbrvnat}
\bibliography{robust-isoform-expression}
\end{document}